\documentclass[12pt]{article}
\usepackage[left=0.5in,right=0.5in,bottom=0.75in,top=0.75in]{geometry}

\usepackage{amsmath}
\usepackage{amsfonts}
\usepackage{comment}
\usepackage{float}
\usepackage{color}
\usepackage[font = small]{caption}
\usepackage{subcaption}
\usepackage{hyperref}
\hypersetup{colorlinks=true, citecolor=black, linkcolor=., urlcolor=cyan}
\usepackage[scr=boondox]{mathalfa}
\allowdisplaybreaks
\usepackage{natbib}
\usepackage{graphicx}
\graphicspath{{img/}}
\usepackage{xcolor}

\usepackage{setspace} 
\usepackage{algpseudocode}
\usepackage[shortlabels]{enumitem}
\numberwithin{equation}{section}
\usepackage{multicol}
\usepackage{ctable}
\setupctable{pos=ht!}

\newcommand{\snr}{\frac{\sigma^2_s}{\sigma^2_\epsilon + \sigma^2_s}}
\newcommand{\stdX}{\dot{\mathbf{X}}}

\newcommand{\rotX}{\tilde{\mathbf{X}}}
\newcommand{\roty}{\tilde{\mathbf{y}}}

\providecommand{\bothvar}{\sigma^2_s + \sigma^2_\epsilon}

\newcommand{\train}{_{-k}}
\newcommand{\test}{_{k}}

\providecommand{\Inv}{^{-1}}

\providecommand{\D}{\mathbf{D}}

\providecommand{\I}{\mathbf{I}}

\providecommand{\K}{\mathbf{K}}

\providecommand{\Q}{\mathbf{Q}}

\renewcommand{\S}{\mathbf{S}}

\providecommand{\U}{\mathbf{U}}
\providecommand{\V}{\mathbf{V}}
\providecommand{\W}{\mathbf{W}}
\providecommand{\X}{\mathbf{X}}

\providecommand{\Z}{\mathbf{Z}}

\providecommand{\x}{}
\renewcommand{\x}{\mathbf{x}}
\newcommand{\y}{\mathbf{y}}

\let\origc\c
\DeclareRobustCommand\c{\ifmmode\mathbf{c}\else\expandafter\origc\fi}
\let\origd\d
\DeclareRobustCommand\d{\ifmmode\mathbf{d}\else\expandafter\origd\fi}
\let\origu\u
\DeclareRobustCommand\u{\ifmmode\mathbf{u}\else\expandafter\origu\fi}
\let\origd\v
\DeclareRobustCommand\v{\ifmmode\mathbf{v}\else\expandafter\origv\fi}

\providecommand{\byy}{\tilde{\mathbf{y}}}

\providecommand{\cL}{\mathcal{L}}

\providecommand{\bb}{\boldsymbol{\beta}}
\providecommand{\bh}{\widehat{\beta}}
\providecommand{\bbh}{\widehat{\boldsymbol{\beta}}}

\providecommand{\bep}{\boldsymbol{\epsilon}}

\providecommand{\bg}{\boldsymbol{\gamma}}

\providecommand{\bS}{\boldsymbol{\Sigma}}

\providecommand{\Var}{\mathbb{V}}

\providecommand{\Tr}{^{\scriptscriptstyle\top}}
\providecommand{\Inv}{^{\scriptscriptstyle-1}}

\providecommand{\one}{\mathbf{1}}
\providecommand{\zero}{\mathbf{0}}

\usepackage{booktabs}
\usepackage{multirow}
\usepackage{tabularx}

\usepackage{amsthm}

\newtheorem*{definition*}{Definition}
\newtheorem{theorem}{Theorem}
\newtheorem*{theorem*}{Theorem}

\newtheorem*{prop*}{Proposition}
\newtheorem{lemma}{Lemma}
\newtheorem*{lemma*}{Lemma}

\newtheorem*{corollary*}{Corollary}

\newtheorem*{conjecture*}{Conjecture}

\makeatletter
\newcommand{\leqnomode}{\tagsleft@true\let\veqno\@@leqno}
\newcommand{\reqnomode}{\tagsleft@false\let\veqno\@@eqno}
\makeatother



\title{Cross-Validation in Penalized Linear Mixed Models: Addressing Common Implementation Pitfalls}
\author{
  Tabitha K. Peter\\Dept. of Biostatistics\\University of Iowa
  \and
  Patrick J. Breheny\\Dept. of Biostatistics\\University of Iowa
}

\begin{document}

\maketitle


\abstract{
In this paper, we develop an implementation of cross-validation for penalized linear mixed models. While these models have been proposed for correlated high-dimensional data, the current literature implicitly assumes that tuning parameter selection procedures developed for independent data will also work well in this context. We argue that such naive assumptions make analysis prone to pitfalls, several of which we will describe. Here we present a correct implementation of cross-validation for penalized linear mixed models, addressing these common pitfalls. We support our methods with mathematical proof, simulation study, and real data analysis. 

  \bigskip

\noindent \textbf{Keywords:} Cross-validation, Penalized Regression, Lasso, High-dimensional data analysis, Linear mixed models
}

\section{Introduction}\label{sec:notation}

Penalized linear mixed modeling (PLMM) is a regression approach designed to analyze correlated high-dimensional data. Penalized regression methods such as the lasso \citep{Tibshirani1996} are attractive for high-dimensional data because they create sparse solutions, and linear mixed modeling is an established framework for analyzing correlated data \citep{Laird1982}. Combining the strengths of these two methodologies has been proposed for analyzing high-dimensional data with correlation structure \citep{Rakitsch2013, ggmix, Reisetter2021a}. Two important areas of potential application for PLMM include genome-wide association studies (GWAS) with population structure and gene expression analyses in the presence of possible batch effects. Recognizing the potential use for PLMM, we also see a need to examine cross-validation implementation for these models. To assume, as most existing literature does, that the tuning parameter selection methods developed for independent data will also work in the PLMM context is naive. We show here that such assumptions make analysis prone to pitfalls, which we address in our development of a cross-validation implementation for PLMMs.  

We begin with defining the $n \times p$ design matrix $\X$, in which the $n$ rows are observations (e.g., samples) and the $p$ columns are features (e.g., genetic variants, etc.). Throughout, we assume that this $\X$ has been column-standardized so that the mean of each column $\x_j$ is $0$ and the variance of each $\x_j$ is $1$ ($j \in 1, ..., p)$. We further assume $\y$ to be an $n \times 1$ column vector representing a normally distributed outcome. We use the data-generating model 
\begin{equation}
\label{eqn:main-model}
    \y = \X\bb + \Z\bg + \bep
\end{equation}
where $\Z$ is a $n \times g$ matrix of indicators representing grouping structures among rows, and $\bg$ is a $g \times \one$ vector representing how the correlation structure in $\Z$ impacts $\y$. Since $\Z$ and $\bg$ are typically unknown in practice, we re-express the model in terms of an unknown confounder $\u$: 
\begin{equation}
\label{eqn:u-model}
    \y = \X\bb + \u + \bep
\end{equation}
under the assumptions that $\bep \perp \u$, $\bep \sim N(0, \sigma^2_\epsilon\I)$, and $\u \sim N(0, \sigma^2_s \K)$.
We further define the  $n \times n$ covariance matrix of $\y$ as $\bS = \sigma^2_s\K + \sigma^2_\epsilon\I$, where $\sigma^2_s$ and $\sigma^2_\epsilon$ represent the variances due to \textbf{s}tructure and noise, respectively. Then we write 
\begin{equation}
\begin{aligned}
     \bS &= [\snr\K + \frac{\sigma^2_\epsilon}{\bothvar}\I](\bothvar) \\
     &= [\eta \K + (1 - \eta\I)]\tau^2,
    \label{eqn:Sigma-def}
\end{aligned}
\end{equation}
where $\eta = \sigma^2_s / (\sigma^2_\epsilon + \sigma^2_s)$ and $\tau^2 = \bothvar$. Note that $\tau^2$ can be absorbed into the penalty parameter $\lambda$ (i.e., this term does not affect the loss).
With these definitions, we may re-express the data generating model as
\begin{equation}
    \label{eqn:generating-mod}
    \y = \X\bb + \u + \bep \equiv \y \sim N(\X\bb, \bS).
\end{equation}
The central aim of a penalized linear mixed model is to precondition (or `rotate') the data as described by \citet{Jia2015}, using the square root of the covariance matrix so that data are \textit{decorrelated}, i.e., 
\begin{equation}
    \label{eqn:rot-mod}
    \bS^{-1/2}\y \sim N(({\bS^{-1/2}}\X)\bb, \bS^{-1/2}\bS\bS^{-1/2}) \equiv \roty \sim N(\rotX\bb, \I).
\end{equation}

To fit such a model, one must estimate $\bS$ which in turn requires estimating $\K$ and $\eta$. Typically, in a PLMM, the correlation among the features, $\hat\K = \frac{1}{p}\X\X^\top$, is used to estimate $\K$ \citep{Hayes2009}. In the specific context of genome-wide association data, $\hat\K$ is also referred to as the genetic relationship matrix or realized relationship matrix. 
To estimate $\eta$, an efficient and straightforward approach is to obtain its MLE under the null model where $\bb = \zero$; note that this is a one-parameter optimization problem \citep{FaSTLMM}. Using $\hat\K$ and $\hat\eta$, we write the estimated covariance matrix $\hat\bS = \hat\eta \hat\K + (1 - \hat\eta) \I$. We will use these estimates in subsequent derivations. 

The rest of this paper is structured as follows: Section \ref{sec:intercept-result} presents the derivation of a result that offers a computationally convenient calculation of the intercept in PLMMs. Section \ref{sec:pitfalls} outlines four pitfalls that arise when applying PLMMs and proposes solutions to these issues. Section \ref{sec:sim-study} presents a simulation study to illustrate the pitfalls described in Section \ref{sec:pitfalls}, and Section \ref{sec:real-data} applies our proposed solutions to real data analysis. The discussion in Section \ref{sec:discussion} makes a generalization about the evidence provided by our simulation study and real data analysis, leaving the reader with practical recommendations on where to begin in using the PLMM to analyze high-dimensional data. 

\section{A closed-form intercept result}\label{sec:intercept-result}

In the case of PLMMs, we find a useful result in which the intercept may be calculated as the mean of the outcome vector $\y$. While it is standard practice in lasso models for the intercept to be calculated as the mean of the outcome, at the outset it is not obvious that such a result can also hold in the correlated context of PLMMs. We take the time to show this result here because of its implications for efficiency in computing PLMMs; this result avoids the need to create a copy of the design matrix that has an intercept column. To avoid such copying becomes particularly advantageous when $p$ is large, as it is in most cases where the PLMM framework would be used. We derive this result using the following two lemmas: 

\begin{lemma}
For the matrix $\hat\K = \tfrac{1}{p}\X\X^\top$, where $\X$ has been column-standardized, all eigenvectors of $\hat\K$ can be partitioned into two categories:
\par 1. Eigenvectors with mean $0$
\par 2. Eigenvectors associated with zero eigenvalues
\end{lemma}

\begin{proof}
    
We write the eigendecomposition $\hat{\mathbf{K}} = \mathbf{U}\mathbf{S}\mathbf{U}^\top$. We denote each column of $\mathbf{U}$ as $\mathbf{u}_k$, $k \in 1, ..., r$, and we denote the eigenvalues of diagonal matrix $\mathbf{S}$ as $s_k$, all of which must be real since $\hat{\mathbf{K}}$ is symmetric. Then we have: 

\begin{align*}
\hat{\K}\u_k &= s_k\u_k \forall k && \text{by definition} \\
\one_n^\top\hat{\K}\u_k &= \one_n^\top s_k\u_k && \text{left multiply} \\
(\one_n^\top\hat{\K})\u_k &= s_k(\one_n^\top\u_k) && \text{associativity} \\
0\u_k &= s_k(\one_n^\top\u_k) && \text{columns of } \hat{\K} \text{ sum to } 0 \text{ because } \one\Tr\X=\zero\\ 
\implies s_k(\one_n^\top\u_k) &= 0 
\end{align*}
Thus, either the eigenvalue is zero or the eigenvector has mean zero.
\end{proof}

\begin{lemma}
    Given a column-standardized $\X$, the inverse of $\bS = \tau^2[\tfrac{\eta}{p} \X\X \Tr + (1-\eta)]\I$ may be written as
    $$ \bS^{-1} = \U_1\Q\U_1 + (1 - \eta)^{-1}\I, $$
    where $\U_1$ is a matrix with mean-0 eigenvectors as its columns and $\Q$ is a matrix of weights. 
\end{lemma}

\begin{proof}
We begin with the definition of $\bS^{-1}$: 

\begin{align*}
     \bS^{-1} &= ([\eta\frac{1}{p}\X\X\Tr + (1 - \hat{\eta})\I]\tau^2)^{-1} && \text{by definition} \\
              &= ([\eta\U\S\U^\top + (1-\eta)\U\U^\top]\tau^2)^{-1} && 
     \text{eigendecomposition; orthogonality} \\
              &= \U([\eta\S + (1 - \eta)\I]\tau^2)^{-1}\U^\top && 
     \text{orthogonality again; factoring} \\
              &= \U\W^2\U^\top && \text{define } \W^2 \text{ as a diagonal matrix of weights}
 \end{align*}
Next, note that we may partition $\U$ and $\W^2$ according to the result in part (a): 
$$
\U = 
\begin{bmatrix}
      \U_1 && \U_2 
\end{bmatrix}
$$

$$
\W^2 = 
\begin{bmatrix}
    \W^2_1 && \zero \\
    \zero && \W^2_2
\end{bmatrix}
$$
where the columns of $\U_1$ each have mean $0$, $\W^2_2 = \text{diag}(\frac{1}{1 - \hat{\eta}})$, and $\W^2$ denotes a block diagonal matrix. We then define $\Q = ([\hat{\eta}\S_r + (1 - \hat{\eta})\I]\tau^2)^{-1}$, where $\S_r$ is the submatrix of $\S$ where $\text{diag}(\S) > 0$ (that is, $\text{diag}(\S_r)$ represents the nonzero eigenvalues of $\K$. Using this $\Q$, we obtain: 
    
\begin{align*}
        \hat{\bS}^1 &= \U_1\W^2_1\U_1^\top + \U_2\W^2_2\U_2^\top && 
        \text{multiply} \\
        \hat{\bS}^{-1} &= \U_1\Q\U_1 + (1 - \hat{\eta})^{-1}\I  && 
        \text{factoring; using Lemma (1) result}
\end{align*}

$\implies \hat{\bS}^{-1}$ may be written as $\hat{\bS}^{-1} = \U_1\Q\U_1 + (1 - \hat{\eta})^{-1}\I.$
\end{proof}

Using Lemmas 1 and 2, we may prove the following theorem: 

\begin{theorem}
\label{theorem:intercept}
The loss of a PLMM may be partitioned into two optimization problems: one part involves only $\beta_0$, and the other involves only $\bb$. 
\end{theorem}

Define the loss of a PLMM as 
\begin{align*}
    \cL &= (\y - \one\beta_0 - \X\bb)\Tr\Sigma\Inv(\y - \one\beta_0 - \X\bb) \\
    &= (\y - \one\bar{y} + \one\bar{y} - \one\beta_0 - \X\bb)\Tr\Sigma\Inv(\y - \one\bar{y} + \one\bar{y} - \one\beta_0 - \X\bb) && \text{add/subtract mean} 
\end{align*}

Then we have the following result from the cross product term: 

\begin{align*}
    \text{cross product} &= (\one\bar{y} - \one\beta_0)\Tr(\U_1\Q\U_1 + (1 - \hat{\eta})^{-1}\I)(\y - \one\bar{y} - \X\bb) && 
    \text{using Lemma (2) result} \\
    &= (\bar{y} - \beta_0)\one\Tr(\U_1\Q\U_1 + (1 - \hat{\eta})^{-1}\I)(\y - \one\bar{y} - \X\bb) && \text{factor} \\
    &= (\bar{y} - \beta_0)\one\Tr(\U_1\Q\U_1 + (1 - \hat{\eta})^{-1}\I)(\dot\y - \X\bb) && 
    \text{let $\dot\y \equiv \y - \one\bar{y}$} \\
    &= (\bar{y} - \beta_0)\one\Tr(1 - \hat{\eta})^{-1}\I(\dot\y - \X\bb) && \one\Tr\U_1 = \zero \\
    &= \frac{(\bar{y} - \beta_0)\one\Tr(\dot\y - \X\bb)}{1 - \hat{\eta}} && \text{simplify} \\
    &= 0 && \one\Tr\dot\y = 0 \hspace{3pt}\text{and}\hspace{3pt} \one\Tr\X = 0
\end{align*}

$\implies \cL$ may be partitioned into two problems, wherein we solve: 
\smallskip
\par 1. $(\bar y - \beta_0)\one\Tr = 0$ (trivial)
\par 2. $\dot\y - \X\bb = 0$ \\
This theorem brings a computational convenience into the application of PLMMs: instead of attaching a $\one$ column to $\X$ and carrying this through the model fitting procedure, we can simply fit a model without an intercept column, and designate $\hat\beta_0 = \bar y$ outside of the model fitting procedure. The beauty of this simplification is especially useful when $p$ is large, as there is no need to create a copy of the design matrix soley for the purpose of adding an intercept column (e.g., a column of $1$s) prior to model fitting. 

\section{Addressing pitfalls in cross-validation}\label{sec:pitfalls}

This section points out and addresses four pitfalls that can arise in implementing cross-validation for PLMMs. Sections \ref{sec:precond} and \ref{sec:low-var-feat} describe mathematical results, whereas Sections \ref{sec:blup} and \ref{sec:rotation} describe computational issues. 

\subsection{Constructing the preconditioning matrix}\label{sec:precond}

In Section~\ref{sec:intercept-result}, we used the spectral decomposition $\K = \U\S\U\Tr$ write $\bS$ in the following form:
\begin{equation}
\label{eqn:sig-def}
    \bS = (\sigma_s^2\K + \sigma_\epsilon^2\I)\tau^2 = [\U (\sigma_s^2 \S + \sigma_\epsilon^2\I) \U \Tr]\tau^2
\end{equation}
To construct the preconditioning matrix, an alternative to the spectral decomposition is to carry out the singular value decomposition (SVD) of  $\X=\U\D\V\Tr$.
This SVD approach would forgo the construction of $\X\X\Tr$; however, the SVD approach causes dimensionality issues. In particular, many SVD implementations use a `thin' SVD, without making the dimensions of $\U$ explicit. If $n > p$, a thin SVD would result in $\U$ having dimension $n \times r$, where $r$ is the number of chosen eigenvectors. This $n \times r$ dimension conflicts with the $n \times n$ dimension of $\I$ term in the variance structure; therefore, the SVD factorization of $\hat\K$ is not compatible with the $n \times n$ unstructured component of the variance. In order to obtain $\U$, we must construct $\hat\K = \frac{1}{p}\X\X\Tr$ and then take the eigendecomposition $\text{eigen}(\hat\K) = \U\S\U\Tr$. 

\subsection{The importance of re-standardization}\label{sec:low-var-feat}

The rotation of $\X$ that yields $\rotX = \bS^{-1/2}\X$ changes the variation in the columns the design matrix. In other words, the variances of the columns of $\rotX$ may be quite different than the variances of the columns of $\X$. This causes problems in penalized regression models, where normalizing these variances is essential to ensure that equal penalization is applied to all features. To maintain this normalization, we re-standardize the data after any maneuver that changes column-wise variation (e.g., preconditioning or subsetting).

For example, suppose that $\X$ has columns $\x_j, j \in 1, \cdots, p$, and suppose that for some $j$ the column $\x_j$ is a feature with low variation, as in 
$$
\X = [\x_1, \x_2, \cdots, \x_j, \cdots, \x_p]
$$
When the rows of $\X$ are partitioned into training set $\X\train$ and test set $\X\test$ in the $k$th fold of of the CV procedure, it is possible that $\x_j$ may become a constant feature in $\X\train$: 
$$
\X\train = [\x_1, \x_2, \cdots, \textcolor{red}{\x_j}, \cdots, \x_p]
$$
where the red color indicates that the feature is constant. Rotating $\X\train$ will transform the columns so that the variance of $\tilde\x_j$ is not exactly zero in $\rotX\train$:
$$
\rotX\train = [\tilde\x_1, \tilde\x_2, \cdots, \textcolor{orange}{\tilde\x_j}, \cdots, \tilde\x_p].
$$
While it is possible to fit a model on $\rotX\train$, this will result in aberrant $\bh_j$ estimates as the optimization will include dividing by the low variation of $\tilde\x_j$. 
We note that this kind of scenario is \textit{probable} to arise in contexts where data include some features with low variation (e.g., genetic markers with low minor allele frequency). This example highlights the importance of re-standardizing after subsetting. To avoid the pitfall in this example, we re-standardize the data $\X\train$ in every fold of cross-validation (when data are subset into test/train sets). Calculating the column-wise variance values acts as a way to `screen' for near-constant features. Any features in $\X\train$ with a zero or near-zero variance are designated to have a penalty of $\infty$, so that those features are never selected for the model that is fit in that given fold. 

\subsection{Prediction for PLMM}\label{sec:blup}
Prediction is an essential element of data analysis, and the PLMM framework lends itself naturally to the use of best linear unbiased prediction (BLUP). BLUP incorporates the correlations among observations in addition to the direct effects of individual features, and this approach increases accuracy in a wide variety of applications \citep{Robinson1991}. The BLUP adjustment is readily obtained from the estimate $\hat\bS$ calculated during the PLMM fitting process. Let $\{\X_1, \y_1\}$ represent a dataset used to fit a penalized linear mixed model, let $\X_2$ represent new data for which predictions are to be made, and partition $\bS = \Var(\y_1, \y_2)$ as: 
$$
\hat\bS = 
\begin{bmatrix}
    \hat\bS_{11} && \hat\bS_{12} \\
    \hat\bS_{21} && \hat\bS_{22}
\end{bmatrix}.
$$
Since the covariance has been estimated in the model fit for $\X_1$, we have $\hat\bS_{11}$ already; we need only to invert this matrix to obtain $\hat\bS_{11}^{-1}$. We also have the residuals from the first model fit, $\y_1 - \X_1\bbh$. Using the new data $\X_2$, all we need to calculate is
$$
\hat\bS_{21} = (\X_2\X_1)^{\Tr}
$$
in order to obtain the BLUP for $\y_2$
$$
\hat\y_2 = \X_2\bbh + \hat\bS_{21}\hat\bS_{11}^{-1}(\y_1 - \X\bbh).
$$
Knowing that the BLUP estimate improves prediction and seeing that this approach is natural in the context of PLMM, we recommend BLUP as the default for prediction with PLMM. 

One important caveat for calculating the BLUP is the need to use consistent scaling for the estimates $\hat\bS_{21}$ and $\hat\bS_{11}$. To highlight this issue, we write the BLUP in terms of $\X_1$ and $\X_2$: 

$$
\hat\y_{\text{BLUP}} = \X_2\bbh + \frac{\hat\eta}{p}\X_2\X_1\Tr\left(\frac{\hat\eta}{p} (\X_1\X_1\Tr) + (1 - \hat\eta)\I\right)^{-1}(\y_1 - \X_1\bbh).
$$
If $\hat\bS_{11}$ was calculated using a column-standardized $\X_1$, then $\X_2$ should be standardized using the same centering/scaling values that were used to standardize $\X_1$ in order to ensure the $\hat\bS_{11}$ and $\hat\bS_{21}$ components are on the same scale. 

This issue of scaling has important implications for cross-validation, which involves fitting a model on the entire data set and then dividing the data into testing/training subsets.
Each fold of cross-validation requires calculating $\X_1\X_1\Tr$, where $\X_1$ denotes the training data for the current fold. In principle, this is the same same as subsetting the covariance from the model fit on the entire dataset, $(\X\X\Tr)_{11}$. However, because the model fitting process involves re-standardizing the design matrix within each fold, these two matrices are different -- one is based on standardizing the columns of $\X$ while the other is based on standardizing the columns of $\X_1$. To be explicit, denote the within-fold re-standardized training and testing datasets as $\stdX_1$ and $\stdX_2$, respectively.
The model being fit within the fold is based on using $\stdX_1\stdX_1\Tr$ to estimate the variance. If the subsets $(\X\X\Tr)_{11}$ and $(\X\X\Tr)_{21}$ are used for the BLUP adjustment, then the estimates of $\bS$ used for BLUP adjustment and used to estimate $\bb$ are different. 
We show in Section \ref{sec:scaling-sim} how this subtle difference in the variance estimates negatively affects estimation.  

\subsection{Rotation in cross-validation}\label{sec:rotation}

The model fitting process consists of three steps: (1) construct the preconditioner $\bS^{-1/2}$ as described in \ref{sec:precond}, (2) rotate the data to obtain $\rotX$, and (3) fit the model on the preconditioned data $\tilde \X, \byy$. This section investigates whether all all of these steps need to be repeated for every cross-validation fold.

In step (1), the preconditioner that must be constructed is $\bS\train^{-1/2}$, where $\bS\train$ denotes the submatrix of $\bS$ consisting of the observations in the fold used for modeling (as opposed to the observations reserved for prediction). The inversion requires an eigendecomposition, which is typically a computationally expensive procedure. Steps (2) and (3) may also take considerable computation time when the dataset is large. We studied three approaches for navigating these computational challenges: 

\begin{enumerate}
    \item \textbf{Full CV}: Carry out all three steps in each fold of CV; this includes taking the eigendecomposition of $\X\train$, rotating $\X\train$ to obtain $\rotX\train$, and fitting the model. 
    \item \textbf{Inner CV}: Takes step (1) outside of the CV procedure; using one eigendecomposition of the entire dataset, simply subset the rows of $\U$ to obtain the preconditioner for $\rotX\train$. We call this inner CV because the preconditioning step happens \textit{inside} each CV fold.
    \item \textbf{Outer CV}: Takes step (2) outside of the CV procedure; rather than rotating the data within each fold to obtain $\rotX\train = \bS\train^{-1/2}\X\train$, this approach preconditions the data a single time and subsets the rows of $\rotX$ to obtain $\rotX\train$. We call this outer rotation because the preconditioning of the data happens \textit{outside} of the CV procedure. 
\end{enumerate}

Table \ref{tab:compare_cv} summarizes these three CV approaches: 

\begin{table}[H]
    \centering
    \begin{tabular}{c|c|c|c}
     \textbf{In each fold} & \textbf{Full} & \textbf{Inner} & \textbf{Outer} \\
     Eigendecomposition  & Yes & No & No \\
     Rotation & Yes & Yes & No \\
     Fit model & Yes & Yes & Yes
    \end{tabular}
    \caption{Comparison of CV approaches}
    \label{tab:compare_cv}
\end{table}

As established by \citet{Hastie2009}, the best practices for CV implementation are to cross-validate each aspect of the model fitting process. From this perspective, full CV is the `gold standard' approach. Moreover, only full CV is able to take advantage of Theorem \ref{theorem:intercept}. Outer CV is the most computationally attractive approach, and inner CV is a compromise between the other two approaches. Sections \ref{sec:sim-study} and Section \ref{sec:real-data} compare these three CV approaches through simulation and real data analysis, respectively. 


\section{Simulation studies}\label{sec:sim-study}

\subsection{Scaling and prediction}\label{sec:scaling-sim}
Our first simulation study is designed to illustrate the pitfall described in Section \ref{sec:blup}, in which using subsets of the full data covariance matrix
$$\hat\bS_{11} = (\X\X\Tr)_{11}, \hat\bS_{21} = (\X\X\Tr)_{21}$$
instead of re-calculating these components on the scale of the standardized training data as 
$$\hat\bS_{11} = \stdX_1\stdX_1\Tr, \hat\bS_{21} = \stdX_2\stdX_1\Tr$$
leads to an inconsistency in our BLUP estimation. In this simulation study, we compared this BLUP implementation (which we refer to as the ``incorrect BLUP'') with two other CV methods for high-dimensional data: \textbf{glmnet}'s \verb|cv.glmnet()| (a penalized approach, but not a mixed model), and \textbf{plmmr}'s \verb|cv_plmm()| (our penalized linear mixed model approach, which has consistent scaling). Figure \ref{fig:bad-blup} illustrates the results, where we see that this misguided shortcut in subsetting BLUP components inflates estimation error. Notice that this would also have important implications for the selection of tuning parameters such as $\lambda$ as well, given that cross-validation is a standard method for tuning parameter selection.

\begin{figure}[H]
    \centering
    \includegraphics[width=0.75\linewidth]{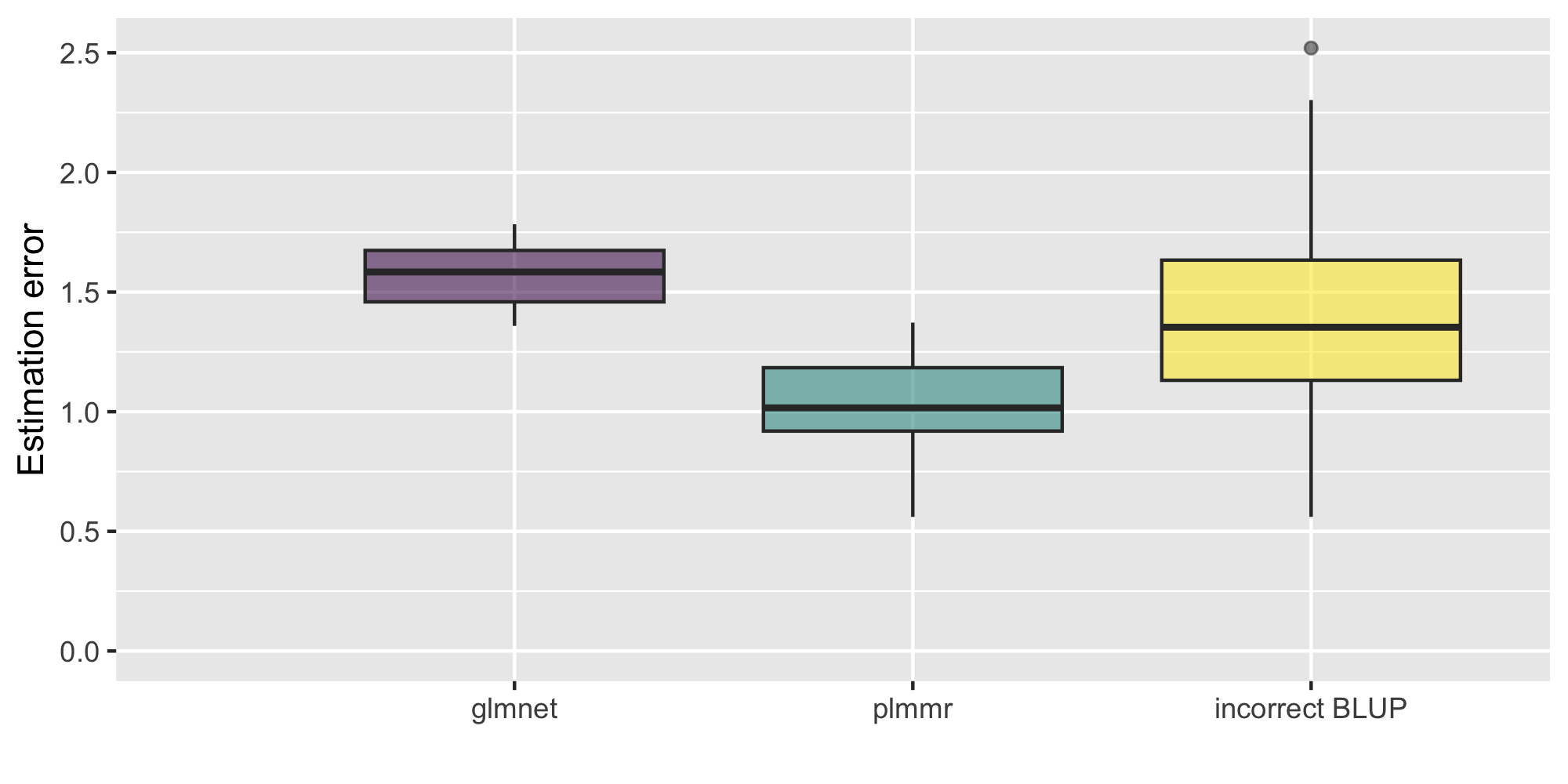}
    \caption{A simulation showing the negative consequences that result from fitting the model and constructing the BLUP using variance estimates on different scales. 30 simulation replications, synthetic correlated data (see Appendix for details). Estimation was assessed using the root-squared estimation error ${||\hat \beta - \beta^* ||}$ (RSEE).}
    \label{fig:bad-blup}
\end{figure}

\subsection{Comparing CV approaches}\label{sec:compare_cv_approaches}
We carried out another simulation study to compare performance of the full, inner, and outer rotation CV techniques. For this simulation study, we used semi-synthetic data in which $\X$ represents real genetic data from 1,401 participants in the PennCath study \citep{PennCath}. To simulate correlation among observations (mimicking phenomena like batch effects or cryptic relatedness), we created a five-level factor variable, assigned each of the 1,401 observations to one of the five levels, and constructed a matrix $\Z$ of indicators corresponding to the factor levels. Having chosen values for $\gamma$ and $\beta$, we simulated a normally distributed outcome $\y$ according to Equation \ref{eqn:main-model}. In all simulation replications, we set the magnitude of $\gamma$ to be $2$. We divided our simulation study into two parts, A and B, based on the magnitudes of the signal $\beta$ values relative to the $\gamma$ parameter. In part A, four $\beta$ were chosen to be the true signals, each having a magnitude of $2$ - we called this the \textit{large signal} setting. In part B, four $\beta$ were chosen to as true signals with a magnitude of $1$ -- this was the \textit{small signal} setting, as $|\beta| < |\gamma|$. 

\subsubsection{Large signal setting} \label{sec:sim-a}

We fit and selected models using each of three CV approaches: full, inner, and outer. Each approach used five CV folds. At the value of $\lambda$ chosen by each approach as the tuning parameter that minimized cross-validation error (CVE), we evaluated several performance metrics, including the false discovery rate (FDR), the number of variables selected (NVAR), the true discovery rate (TDR), the CVE, and the RSEE. Both Table \ref{tab:large_scale_metrics} and Figure \ref{fig:large-signal-sim} summarize these metrics.

We notice from Table \ref{tab:large_scale_metrics} and Figure \ref{fig:large-signal-sim} that the outer CV approach performs notably worse than the other approaches across all performance metrics. Outer CV chooses over 200 variables on average, with a high false discovery rate, and has an estimation error almost twice as high as the other approaches.

Full and inner CV are more comparable, with nearly identical estimation error. Full CV does have the advantage of a lower average false discovery rate compared to inner CV, 0.64 compared to 0.77. Corresponding to these FDR results, the inner CV approach chose a somewhat larger model than full CV.   

\begin{figure}[H]
    \centering
    \includegraphics[width=0.9\linewidth]{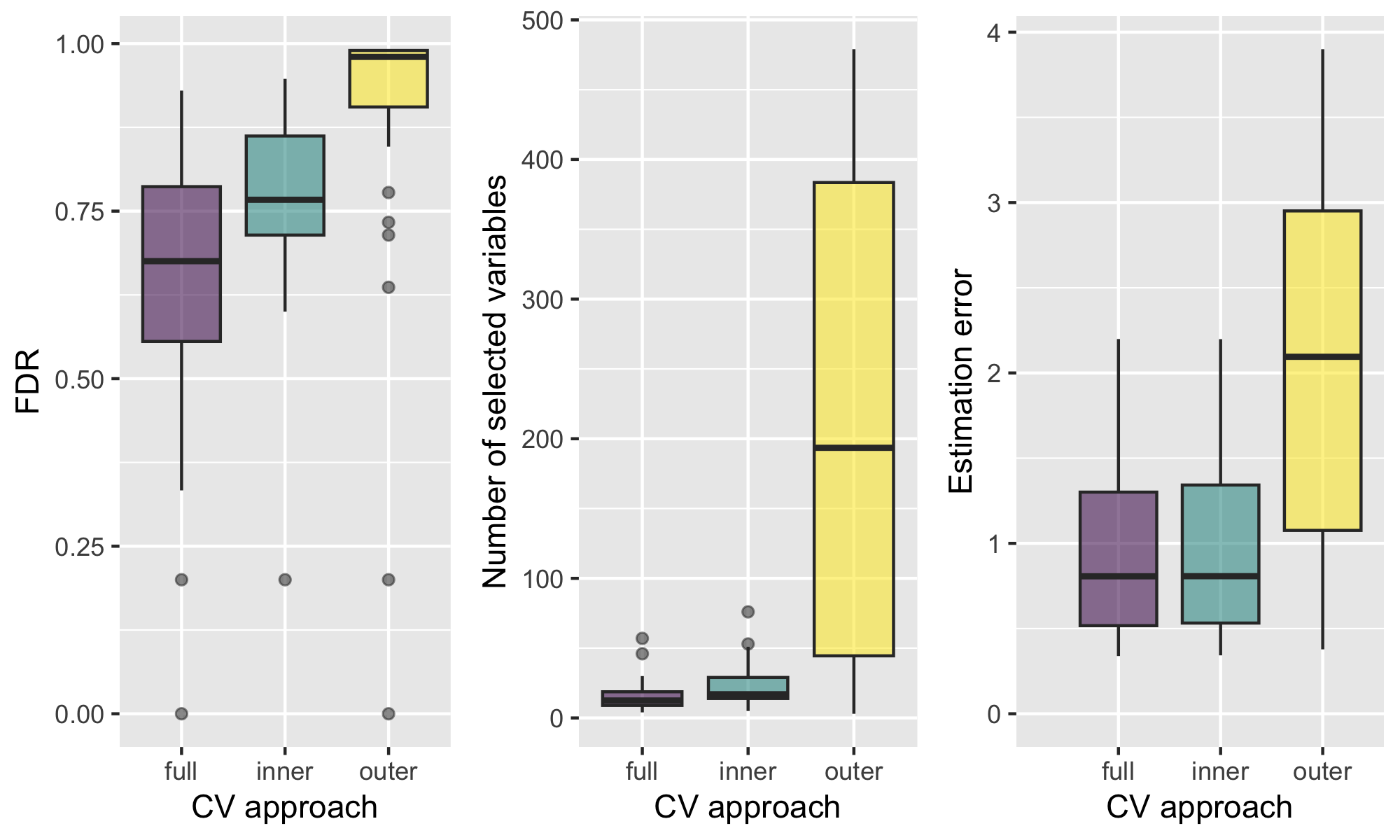}
    \caption{Large signal simulation ($|\beta| = |\gamma|$)  with 30 replications. Estimation was assessed using the root-squared estimation error ${||\hat \beta - \beta^* ||}$ (RSEE).}
    \label{fig:large-signal-sim}
\end{figure}

\begin{table}[H]
    \caption{Large signal simulation metrics}
    \fontsize{12.0pt}{14.4pt}\selectfont
    \begin{tabular*}{\linewidth}{@{\extracolsep{\fill}}lccc}
        \toprule
        & \textbf{Full} & \textbf{Inner} & \textbf{Outer} \\
        \midrule
        \addlinespace[2.5pt]
        \bfseries TDR & 0.98 (0.08) & 0.98 (0.08) & 0.98 (0.08) \\ 
        \bfseries FDR & 0.64 (0.21) & 0.77 (0.14) & 0.88 (0.23) \\ 
        \bfseries NVAR & 16 (12) & 23 (16) & 215 (172) \\ 
        \bfseries RSEE & 0.98 (0.56) & 0.96 (0.55) & 2.04 (1.12) \\
        \bottomrule
    \end{tabular*}
    \begin{minipage}{\linewidth}
        Format: Mean (SD); N. simulation replicates = 30
    \end{minipage}
    \label{tab:large_scale_metrics}
\end{table}

\subsubsection{Small signal setting} \label{sec:sim-b}

Table \ref{tab:small_scale_metrics} and Figure \ref{fig:small-signal-sim} show results from the simulations where the confounding was of greater magnitude than the signal. The FDR and model size results from this setting show an even more pronounced problem with the outer CV approach, which chose over 600 variables on average and maintained an FDR of about 0.99 in all replications. Regarding estimation error, we see that outer CV performs much worse than the other two approaches, just as we saw in the large signal setting. Here again, RSEE is comparable between the full and inner CV approaches; the distinguishing factor between full and inner CV is in the FDR and NVAR metrics. The TDR is a little lower across all of these approaches compared to the large signal case. 


\begin{figure}[H]
    \centering
    \includegraphics[width=0.9\linewidth]{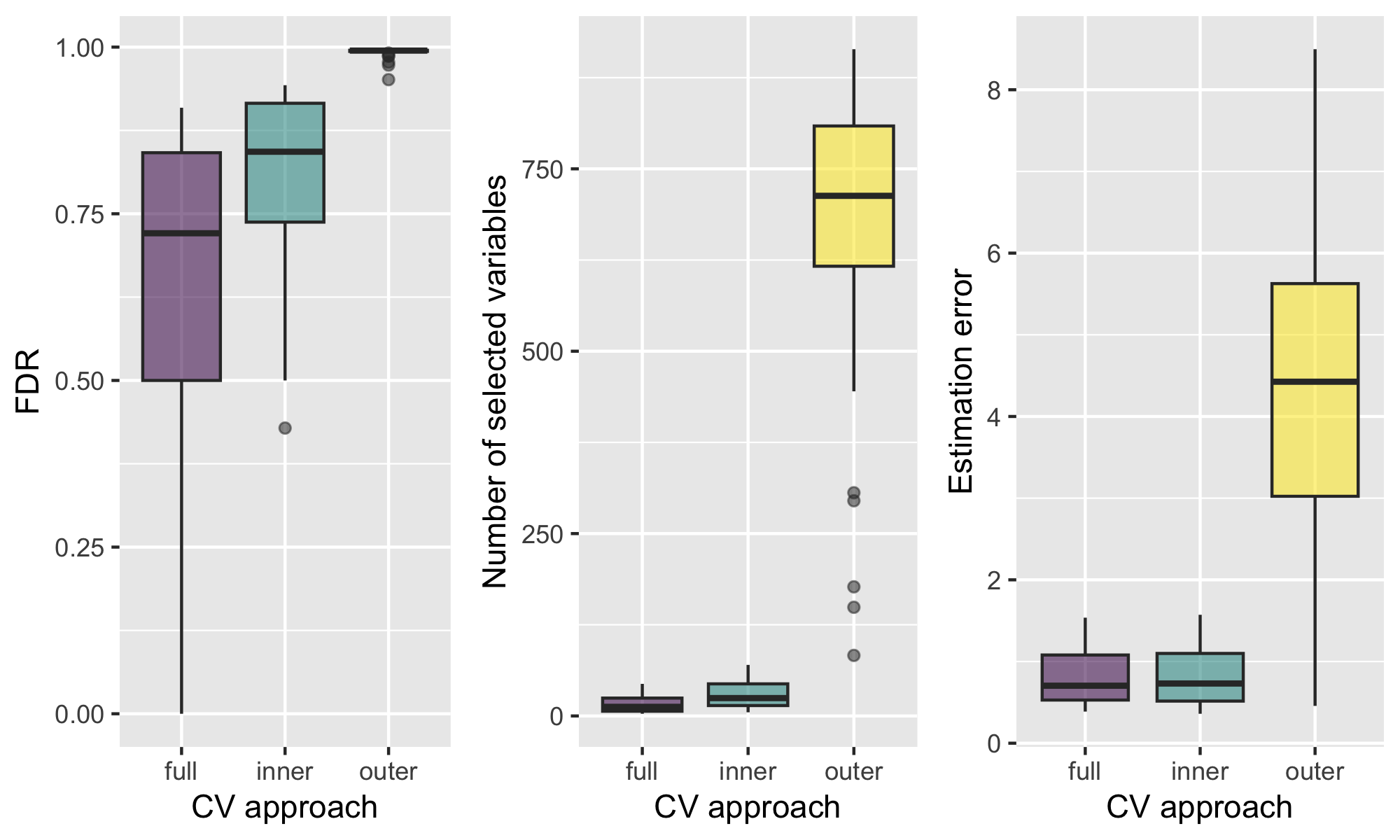}
    \caption{Small signal simulation ($|\beta| < |\gamma|$) with 30 replications. Again, estimation was assessed using the RSEE ${||\hat \beta - \beta^* ||}$.}
    \label{fig:small-signal-sim}
\end{figure}

\begin{table}[!t]
    \caption{Small signal simulation metrics}
    \fontsize{12.0pt}{14.4pt}\selectfont
    \begin{tabular*}{\linewidth}{@{\extracolsep{\fill}}lccc}
        \toprule
        & \textbf{Full} & \textbf{Inner} & \textbf{Outer} \\
        \midrule
        \addlinespace[2.5pt]
        \bfseries TDR & 0.92 (0.15) & 0.92 (0.15) & 0.95 (0.14) \\
        \bfseries FDR & 0.63 (0.25) & 0.80 (0.14) & 0.99 (0.01) \\ 
        \bfseries NVAR & 16 (11) & 29 (19) & 648 (234) \\
        \bfseries RSEE & 0.81 (0.38) & 0.84 (0.39) & 4.30 (1.91) \\
        \bottomrule
    \end{tabular*}
    \begin{minipage}{\linewidth}
        Format: Mean (SD); N. simulation replicates = 30
    \end{minipage}
    \label{tab:small_scale_metrics}
\end{table}

As a generalization across the results of both the small and large signal simulation settings, evidence from the given metrics shows that full CV is the best approach for selecting a model both in terms of having the lowest FDR, lowest prediction error, and most accurate estimation. 

\subsubsection{Examining true prediction error across CV approaches}\label{sec:examine_pred_err}

Another important aspect of evaluating a cross-validation approach is to examine whether the cross-validation error is in fact a good estimate of the true prediction error. In order to assess the true prediction error (that is, prediction error outside of CV) of the three candidate CV approaches, we created testing and training data sets using a partition of the observations in the data used for the above-described simulations. Figure \ref{fig:cve_v_mspe_sims} compares CVE and MSPE for each CV approach across both simulation settings, with reference lines (dotted) showing the slope for CVE = MSPE (the ideal). 

The plots for outer CV illustrate that this approach severely overestimated CVE in the large signal setting, and severely underestimated CVE in the small signal setting. We conclude from this evidence that the outer CV approach is unable to accurately estimate prediction error. Meanwhile, full and inner CV performed comparably in the large signal setting, with CVE and MSPE aligning well.  In the small signal case, CVE was a little lower for full and inner CV than their MSPE -- this reflects the difficulty of estimating error with accuracy when the confounding `drowns out' the relatively small signal. 

\begin{figure}[H]
    \centering
    \includegraphics[width=0.85\linewidth]{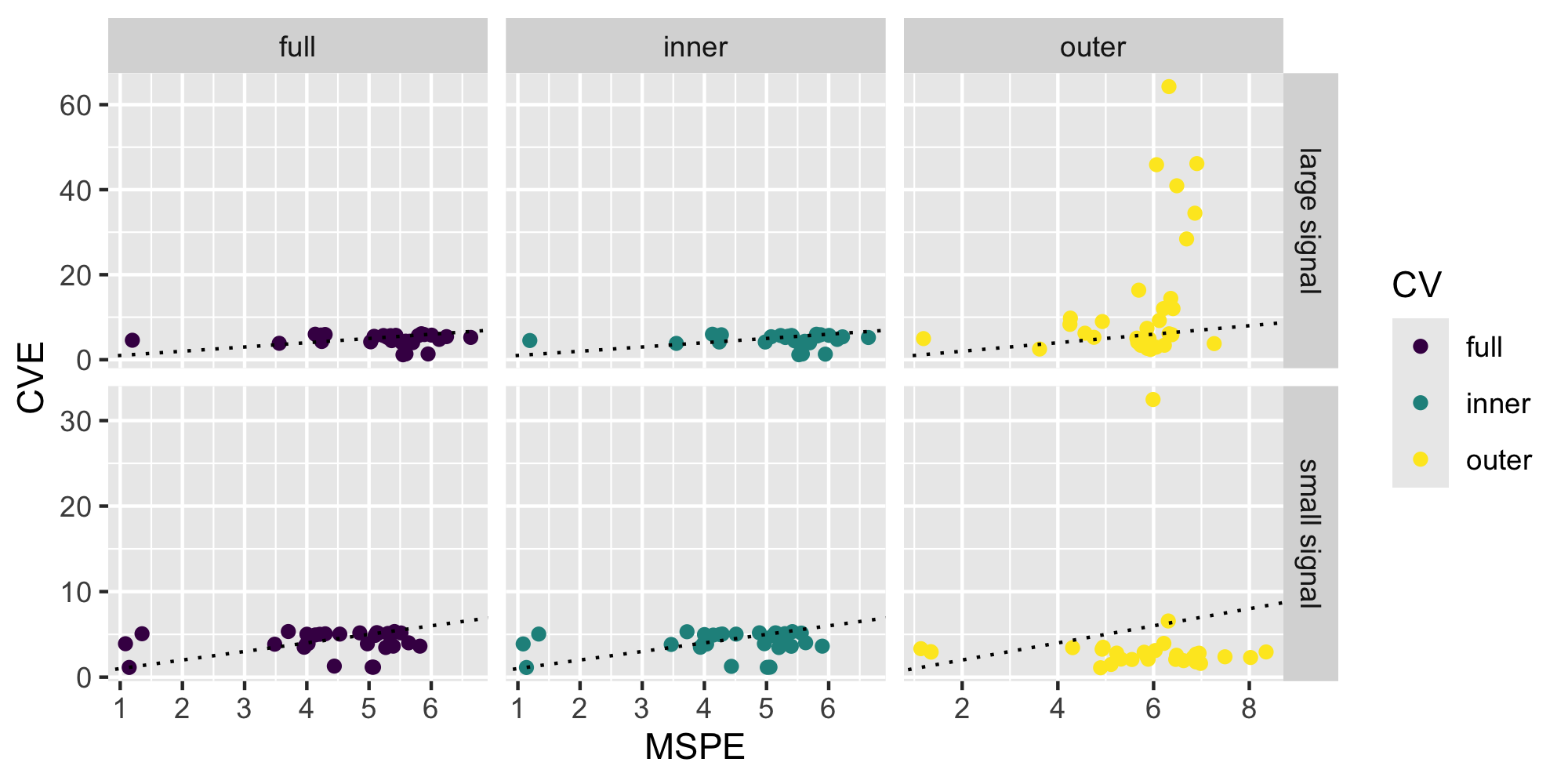}
    \caption{Comparison of cross-validation error (CVE) and mean-squared prediction error (MSPE) from semi-synthetic data. A set of 400 observations was held out as a test set, while the remaining observations were used to fit a model using each of the three candidate CV approaches. As in the simulations described above, we generated a synthetic outcome variable with one confounder. Using the models chosen by each CV approach, prediction error was assessed using the MSPE $||\y_{\text{test}} - \hat\y||^2/n$. We repeated this simulation for 30 replications.}
    \label{fig:cve_v_mspe_sims}
\end{figure}

\section{Real data analysis example}\label{sec:real-data}

In this section, we analyze a subset of the `PennCath' genetic data first mentioned in Section~\ref{sec:compare_cv_approaches}. Unlike the simulations, in this real data analysis the true correlation structure is unknown. We included sex, age, and 4,307 SNPs as predictors, with the presence/absence of coronary artery disease as the outcome (treated as a numeric value 0/1). 

We fit a model with each of the three CV approaches, as summarized in Table \ref{tab:cv-model-size} and in Figure \ref{fig:real_data_err_compare}. These results indicate that the outer CV approach suffers from severe overfitting, as is evidenced by the model size. Moreover, Figure \ref{fig:real_data_err_compare} highlights that the CVE calculated by outer CV is quite different from the true MSPE, indicating that there is a fundamental problem with outer CV as an approach.  

The comparison between inner and full CV is more nuanced, as these two methods performed comparably in terms of model size. Figure \ref{fig:real_data_err_compare} illustrates that while inner CV estimates MSPE accurately overall, it underestimates prediction error for larger values of $\lambda$. By contrast, full CV estimates prediction error accurately along the entire solution path, making it the only fully robust approach to cross-validation.

\begin{table}[H]
\caption{Comparison of model size from real data analysis}
\centering
\begin{tabular}[t]{lrrrrrr}
\toprule
\multicolumn{1}{c}{} & \multicolumn{2}{c}{\textbf{Min}} & \multicolumn{2}{c}{\textbf{1se}} \\
\cmidrule(l{3pt}r{3pt}){2-3} \cmidrule(l{3pt}r{3pt}){4-5}
\textbf{CV} & $\lambda_\text{min}$ & NVAR & $\lambda_\text{1se}$ & NVAR\\
\midrule
Full & 0.0437 & 5 &  0.0753 & 2 & \\
Inner & 0.0387 & 9 & 0.0667 & 2 & \\
Outer & 0.0105 & 556 & 0.0134 & 451 & \\
\bottomrule
\end{tabular}
\label{tab:cv-model-size}
\end{table}

\hspace{5em}

\begin{figure}[H]
    \centering
    \includegraphics[width=0.95\linewidth]{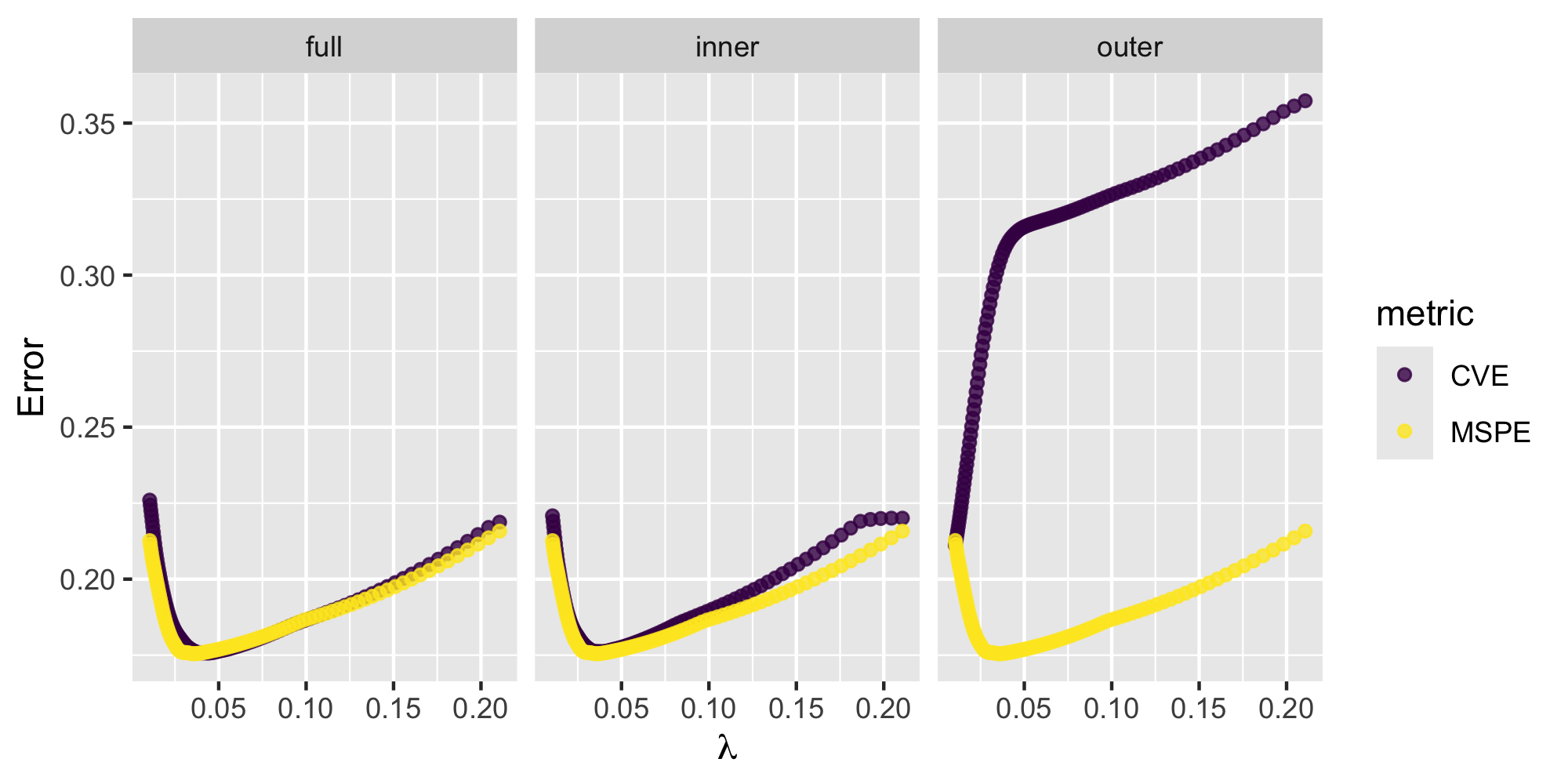}
    \caption{Comparison of CVE and MSPE from real data. 1,001 observations were used to select a model with each CV approach. The remaining 400 observations were held out of model fitting and were used to assess the prediction performance of the selected models. CVE and MSPE were measured in the same way as in Figure \ref{fig:cve_v_mspe_sims}.}
    \label{fig:real_data_err_compare}
\end{figure}

\newpage

\section{Discussion}\label{sec:discussion}

Exchangeability is a fundamental issue for any preconditioned modeling approach, of which penalized linear mixed modeling is a specific example. The outer rotation method fails because it violates exchangeability; that is, once the data $\X$ have been rotated to obtain $\rotX$, the rows of $\rotX$ are no longer exchangeable. Although the original data $\X$ has correlation between rows, each row contains the same amount of information.  \citet{Rabinowicz2022} have pointed out that this exchangeability between rows is the key to implementing CV with correlated data. However, when we precondition $\X$ by $\bS^{-1/2}$ this exchangeability no longer holds, as $\bS^{-1/2}$ weights the observations of $\X$. Typically, $\S$ represents the eigenvalues of $\K$ as shown in \eqref{eqn:sig-def}; after we precondition, the rows of $\rotX$ corresponding to the largest eigenvalue will contain the most information -- and some rows of $\rotX$ might have zero weight. Unlike the rows of $\X$, the rows of $\rotX$ have different amounts of information about the outcome. Therefore, preconditioning renders the rows of $\tilde{\X}$ inexchangeable. Our simulation studies and real data analysis show that this results in severe overfitting.


While inner CV avoids this exchangeability issue, subsetting the $\U$, $\S$, and $\hat\bS$ from the full dataset introduces data leakage \citep{Kaufman2012, Kapoor2023}. Unlike outer CV, inner CV produces mostly reasonable models. However, our investigations using simulated and real data show that it does not always estimate true prediction error accurately, leading to an inflated false discovery rate.

In conclusion, the only fully robust approach to model selection is full cross-validation. This result is of particular importance given that without dedicated software for fitting preconditioned models (including penalized linear mixed models), outer CV is exactly what analysts using off-the-shelf software would be doing unless they programmed the entire CV procedure themselves. Our \textbf{plmmr} package \citep{plmmr} ensures that users carry out full CV for penalized linear mixed models in a manner that is both computationally efficient and methodologically sound.


\bibliographystyle{ims-nourl}

\section*{Appendix} 

The code used to generate data for the simulation in Section \ref{sec:blup} was as follows: 

\begin{verbatim}
#' A function to simulate correlated data
#'
#' @param n An integer with the number of observations
#' @param p An integer with the number of features
#' @param s An integer with the number of signals
#' @param gamma A number representing the bound of the magnitude of the confounding
#' @param beta A number indicating the magnitude of the signal
#' @param B An integer indicating the number of batches
#' @param dat
#'
#' @returns A list with:
#' * y (the outcome vector)
#' * X (the design matrix)
#' * Z (the matrix with the true grouping structure of the batch membership)
#' * gamma (the gamma vector)
#' * id (the vector indicating batch assignments)
#' @export
#'
generate_correlated_data <- function(n=100, p=256, s=4, gamma=6, beta=2, B=20) {
    mu <- matrix(rnorm(B*p), B, p)
    z <- rep(1:B, each=n/B)
    X <- matrix(rnorm(n*p), n, p) + mu[z,]
    b <- rep(c(beta, 0), c(s, p-s))
    g <- seq(-gamma, gamma, length=B)
    y <- X %*% b + g[z]
    Z <- model.matrix(~0+factor(z))
    list(y=y, X=X, beta=b, Z=Z, gamma=g, mu=mu, id=z)
}
\end{verbatim}

\end{document}